\documentclass[11pt]{article}

\usepackage{amsthm,amssymb,amsfonts,amsmath}
\usepackage{amscd} 
\usepackage{mathrsfs}
\usepackage{graphicx}
\usepackage{color}
\usepackage{hyperref}
\usepackage{bm}
\usepackage{algorithm}
\usepackage{algorithmic}
\usepackage[margin=1.25in]{geometry}
\graphicspath{ {./Images/} }
\usepackage{authblk}
\usepackage{todonotes}

\usepackage[backend=bibtex,style=ieee-alphabetic,natbib=true]{biblatex} 
\newcommand{\bra}[1]{\left(#1\right)}
\newcommand{\abs}[1]{\left|#1\right|}
\newcommand{\sbra}[1]{\left[#1\right]}
\newcommand{\set}[1]{\left\{#1\right\}}

\let\epsilon=\varepsilon
\let\rho=\varrho

\newtheorem{theorem}{Theorem}
\newtheorem{lemma}[theorem]{Lemma}
\newtheorem{prop}[theorem]{Proposition}
\newtheorem{corollary}[theorem]{Corollary}
\newtheorem{definition}[theorem]{Definition}

\numberwithin{theorem}{section}
\numberwithin{equation}{section}

\newcommand\cS{{\mathcal S}}

\newcommand{\bP}{\mathbb{P}}

\newcommand{\bZ}{\mathbb{Z}}

\begin{document}
\title{Perfect Sampling of $q$-Spin Systems on $\mathbb Z^2$\\ via Weak Spatial Mixing}
\author{Konrad Anand\thanks{Supported by a studentship from the School of Mathematical Sciences.} }
\author{Mark Jerrum\thanks{Supported by grant EP/S016694/1 `Sampling in hereditary classes' from the Engineering and Physical Sciences Research Council (EPSRC) of the UK.}}
\affil{School of Mathematical Sciences\\Queen Mary, University of London\\Mile End Road, London E1 4NS\\United Kingdom}
\date{}
\maketitle

\begin{abstract}
We present a perfect marginal sampler of the unique Gibbs measure of a spin system on $\bZ^2$. The algorithm is an adaptation of a previous `lazy depth-first' approach by the authors, but relaxes the requirement of strong spatial mixing to weak.  It exploits a classical result in statistical physics relating weak spatial mixing on $\bZ^2$ to strong spatial mixing on squares. When the spin system exhibits weak spatial mixing, the run-time of our sampler is linear in the size of sample. Applications of note are the ferromagnetic Potts model at supercritical temperatures, and the ferromagnetic Ising model with consistent non-zero external field at any non-zero temperature.
\end{abstract}

\section{Introduction}

This paper presents a new version of `lazy depth-first sampling', as presented in \cite{AnandJerrum2021}, which extends the efficacy of the algorithm beyond the domain of strong spatial mixing.

Originally developed in the study of statistical mechanics, spin systems are now also an active area of research in probability, machine learning, and theoretical computer science. Different communities study the same basic concept using different terminology, an example being Constraint Satisfaction Problems (CSPs), which originated in artificial intelligence. In brief, a spin system is defined by a graph~$G$, a finite (in our case) set of spins, and a local `specification'.  A configuration of the system is an assignment of spins to the vertices of the graph.  The local specification defines a probability distribution on the configurations of the system.  Although the specification is local, this Gibbs distribution can exhibit long-range effects.  A primary algorithmic goal in this area is to sample configurations from the Gibbs distribution.

The question of sampling from spin systems has produced an array of results recently and over the past decades. Many approaches have been based in Markov chains, in particular using Glauber dynamics \cite{glauber1963}. Another direction is that of algorithms which try to reduce randomness by computing or partially computing the proportions of the samples. These results have been come from decay of correlations \cite{weitz2006counting}, Taylor expansion of the partition function \cite{PatelRegts, BarvinokSoberon}, and efficiently computing a Markov chain's history \cite{Feng2022TowardsDM}.

The strand of research which this paper falls into is perfect sampling. One motivation is obvious: a perfect sampler has no error, but another less obvious motivation is having a definite termination condition, in contrast to Markov chain simulation which has to be carried on for `sufficiently many' steps. This requires an a priori bound on the mixing time of a Markov chain, which must be analytically derived and may be much larger than actually necessary. In contrast, perfect samplers can respond to characteristics of a problem  instance, and produce results more quickly on favourable instances.

The first demonstration of perfect sampling came with the Coupling From The Past (CFTP) approach of Propp and Wilson \cite{ProppWilson}.  There have been other recent approaches to the problem of efficient perfect sampling which move away from Markov chains, such as the partial rejection sampling of Guo, Jerrum and Liu \cite{guo2019uniform} and the Bayes filter of Feng, Guo, and Yin \cite{Feng2019PerfectSF}. The original lazy depth-first sampling \cite{AnandJerrum2021} and the adapted version introduced in this paper fall in a similar direction.

Another benefit of lazy depth-first sampling is its applicability to infinite spin systems. Spin systems are defined in a local sense, but under certain conditions, which we assume to hold in this paper, the local definition of the spin system extends to a unique Gibbs measure. Of course we cannot sample a configuration on an infinite graph, but a lazy depth-first sampler is intrinsically a marginal sampler, so we can sample finite windows into it.

Van den Berg and Steif \cite{vdBerg1999codings} pioneered perfect sampling from infinite Gibbs measures. They showed that perfect sampling is possible in the case of the ferromagnetic Ising model on $\mathbb{Z}^2$ in the uniqueness regime, i.e., when there is a unique Gibbs measure, using  CFTP. Spinka~\cite{spinka2020finitary}, in a wide-ranging investigation, has taken forward this approach. In contrast, lazy depth-first sampling departs from CFTP, and takes inspiration from partial rejection sampling and related work.

\subsection{Lazy Depth-First Sampling}

Lazy depth-first sampling strengths are clear: it is a perfect sampler, it is a marginal sampler which can sample windows into infinite spin systems, and it is efficient (linear time) on any spin system which exhibits strong spatial mixing with the added condition of subexponential growth of the underlying graph $G$. One can also view it as a sublinear marginal sampler in terms of the size of the whole graph. This is a particularly nice benefit of the algorithm, as most approaches have run-time dependent on the size of the graph, rather than the size of the window.

Lazy depth-first sampling is a simple, recursive algorithm called on a single vertex, $v \in G$. Briefly, we try to guess the spin at a vertex $v$ without looking at spins that live on a boundary around $v$ by computing as much of $v$'s distribution as possible. When we cannot guess, the algorithm recursively determines the spins on a sphere of suitable radius before returning to assign a spin to~$v$. As it runs, the algorithm builds a branching process of recursive calls whose expected size is finite with probability 1 when strong spatial mixing holds.

Relevant conditions for efficiency of lazy depth-first sampling (and many other algorithms) are weak and strong spatial mixing. Informally, weak spatial mixing requires that the marginal distribution of a set of vertices changes little if we change far away spins and nearby spins are unconstrained. Strong spatial mixing requires that the marginal distribution changes little if we change far away spins even when some local spins are fixed.  (We sometimes call the arrangement of fixed spins a `partial configuration'.)

In our first presentation of lazy depth-first sampling to spin systems, our sampling procedure did not take into account the precise structure of the underlying graph and relied on strong spatial mixing to provide a coarse bound to ensure efficiency. In many approaches \cite{Feng2019PerfectSF, spinka2020finitary, sinclair2013spatial,sinclair2017spatial, MR2060631}, this condition is required for efficiency, but some approaches require less: weak spatial mixing \cite{Chen2020OptimalMO, chen2020rapid}. Our adaptation pushes past strong spatial mixing and requires only weak spatial mixing.

The inspiration for the adaption is the classical physics result by Martinelli, Olivieri and Schonnman showing that, for spin systems on $\bZ^2$, weak spatial mixing implies strong spatial mixing on squares \cite{Martinelli1994weakstrong}. In standard strong spatial mixing, the partial configurations are freely arranged, while in strong spatial mixing on squares the partial configurations are restricted to live on a square around the vertices whose change in marginals we are bounding. We use similar arguments to show the efficiency of this adaptation.

This adaptation runs on $\bZ^2$ and takes care to live on a mesh (sublattice of $\bZ^2$) so that all our partial configurations `mostly' live on a square around the vertex we are currently sampling. By doing so, we bypass pathological partial configurations and ensure efficiency of the algorithm with only weak spatial mixing. The algorithm also holds for $q$-spin systems, rather than only two spin systems, and it is worth remarking that one can show efficiency of the algorithm on a variety of 2-d lattices---we consider only $\bZ^2$ for simplicity.

It is worth remarking that the framework of lazy sampling is not restricted to spin systems. Indeed, two recent papers by He, Wang, and Yin and He, Wu, and Yang have applied these principles to the sampling uniform solutions to constraint satisfaction problems \cite{he2022sampling, he2022improved}. Their results show the potential in applications of lazy sampling beyond spin systems.

Informally (with mild added conditions) our result is this:
\begin{theorem}
    Let $\cS$ be a $q$-spin system on $\bZ^2$ which exhibits weak spatial mixing, and let $S$ be as above. Then \textsc{lazy} (Algorithm \ref{alg:LDFS}) is a perfect sampler for $\cS$ which runs in time linear in the number of vertices sampled.
\end{theorem}

\subsection{Applications}

We focus on two new applications for perfect sampling of spin systems on $\bZ^2$: the ferromagnetic Potts model at supercritical temperature and the ferromagnetic Ising model with constant external field at all temperatures. Each of these results leverages previous work on these models.

Sampling from the ferromagnetic Potts model is an active area of research. There are past approximate samplers at supercritical temperatures for general graphs, e.g., Bordewich, Greenhill, and Patel's work on the mixing of Glauber dynamics for maximum degree $\Delta$ graphs \cite{vireshgreenhillbordewich}, approximate samplers specifically on $\mathbb Z^2$ \cite{MR4060355, MR3794520, MR3663632, MR3068036}, approximate samplers for supercritical temperatures \cite{barvinokregts19, helmuthperkinsregts20, vireshguuspotts}, and a recent sampler of the ferromagnetic Potts model at all temperatures for large enough $q$ \cite{MR4141796}.

We present the first perfect sampler for the ferromagnetic Potts model on $\bZ^2$ at supercritical temperatures that works in constant time per vertex. To guarantee the success and efficiency of the Algorithm \ref{alg:LDFS}, we rely on the celebrated work of Beffara and Duminal-Copin \cite{MR2948685} in conjunction with the work of Alexander \cite{MR1626951}.

A major early sampling step for sampling the ferromagnetic Ising model with external field was Mossel and Sly's work on the rapid mixing of the Glauber dynamics provided that the temperature is bounded by a function of the maximum degree of the graph \cite{MosselSly13}. More recent developments establish rapid mixing of the Swendsen-Wang dynamics on stochastic partitioned graphs \cite{Park2017RapidMS} as well as a CFTP algorithm \cite{Feng2022SamplingFT} and a near linear time field dynamics algorithm \cite{MR4538128} for  sampling on bounded degree graphs with the external field bounded away from 1.

We present a perfect marginal sampler for the ferromagnetic Ising model on $\bZ^2$ with consistent external field that works in constant time per vertex with external field bounded away from 1. While this algorithm does not operate on general bounded degree graphs, it is perfect on an infinite graph and is exactly linear in the size of the window. The efficiency of the algorithm comes from the weak spatial mixing condition established by Schonmann and Schlosman \cite{SchonmannShlosman}.

\section{Spin systems}
We define spin systems on finite graphs; we will treat their extension to infinite graphs later.

Given a finite graph $G = (V,E)$ and $q \in \mathbb N$ we work with the set of configurations
$\Omega_V = [q]^V$, where $[q]=\{1,2,\ldots,q\}$.  When there is no room for confusion, we will write $\Omega$ instead of~$\Omega_V$.  Naturally, when $W\subset V$ we write $\Omega_W=[q]^W$ for the partial configurations restricted to~$W$.

For $\sigma \in \Omega$ and $W \subset V$ we denote the restriction of $\sigma$ to $W$ by~$\sigma_W$. For simplicity, when $W= \{v\}$ we will write $\sigma_v$ instead of $\sigma_{\{v\}}$. In this same situation, we denote the spin at $v$ by $i$, and write $\sigma_v=i$ rather than $\sigma_v=\{(v,i)\}$. For two configurations $\tau$ on $\Lambda_1$ and $\sigma$ on $\Lambda_2$ which agree on $\Lambda_1 \cap \Lambda_2$ we denote by $\tau \oplus \sigma$ the configuration on $\Lambda_1 \cup \Lambda_2$ which agrees with both $\tau$ and $\sigma$.

A spin system is defined by a `field' $b : [q] \to \mathbb R$ and a symmetric matrix $A: [q] \times [q] \to \mathbb R$ of interaction weights.  Given a graph~$G$ (finite at this stage) the Gibbs distribution for $G$ gives to each configuration~$\sigma$ the weight
\begin{align*}
  \prod_{v \in V} b (\sigma_v) \prod_{(u,v) \in E} A(\sigma_u,\sigma_v).
\end{align*}
Define the partition function $Z(G)$ to be the sum of all possible weights
\begin{align*}
  Z(G) := \sum_{\sigma \in \Omega} \prod_{v \in V} b (\sigma_v) \prod_{(u,v) \in E} A(\sigma_u,\sigma_v).
\end{align*}
Then we define the Gibbs measure to be the measure $\mu : \Omega \to [0,1]$ where
\begin{align}\label{eq:specification}
  \mu(\sigma) = \frac{\prod_{v \in V} b (\sigma_v) \prod_{(u,v) \in E} A(\sigma_u,\sigma_v)}{Z(G)}.
\end{align}
Frequently in this paper we will want to look at the marginal distribution induced on a set of vertices by fixing the configuration on the remainder of the vertices. For any $W$, define the marginal distribution of $\mu$ on $W$ by 
$$
\mu_W(\sigma)=\sum_{\sigma'\in\Omega:\sigma'_W=\sigma}\mu(\sigma'),
$$
for all $\sigma \in \Omega_W$.  Also, for any $\Lambda \subset V$ and $\tau \in \Omega_\Lambda$, define the marginal distribution with boundary condition $(\Lambda,\tau)$ by 
\begin{align}\label{eq:conditional}
  \mu_W^{(\Lambda,\tau)}(\sigma) :=
     \frac1{\mu_{\Lambda}(\tau)} \sum_{\substack{\sigma'\in\Omega:\\ \sigma'_\Lambda = \tau, \sigma'_W = \sigma}}\!\!\mu(\sigma'),
\end{align}
assuming $\mu_\Lambda(\tau)\not=0$.
We abbreviate $\mu_W^{(\Lambda,\tau)}$ to $\mu_W^\tau$ when the set $\Lambda$ is clear from the context.

With this definition of the Gibbs measure, we have a fixed probability space determined by a few parameters. We call this a spin system $\mathcal S = (G, q, b, A)$. In this paper we deal exclusively with spin systems in which the underlying graph~$G$ is the 2-D grid $\mathbb{Z}^2$, i.e., the graph with vertex and edge sets
\begin{align*}
    V_{\bZ^2} &= \bZ^2, \\
    E_{\bZ^2} &= \set{(u,v) \in V_{\bZ^2} : d(u,v) = 1},
\end{align*}
where $d(\cdot,\cdot)$ is Euclidean distance.
The restriction to the particular 2-D grid is not strictly necessary: with mild modifications to the sampling patterns and analysis the algorithm works on many planar lattices, but for clarity we have restricted our discussion to the most common planar lattice.  However, we shall see that the restriction to two dimensions is crucial.  

Fortunately, when describing and analysing the proposed algorithm, we only encounter finite subgraphs of $\bZ^2$. For these, the Gibbs distribution on configurations is defined by~(\ref{eq:specification}). As usual, the \emph{Gibbs property} of $\mu$ is important to us, and is the reason we can concentrate on finite subgraphs.  Suppose $W\subset V$ and that the configuration $\tau\in\Omega_{V\setminus W}$ is \emph{feasible}, i.e., $\mu_{V\setminus W}(\tau)>0$.  Define the \emph{boundary} of~$W$ by 
    $$
    \partial W=\{v\in V\setminus W:\{u,v\}\in E\text{ for some }u\in W\}, 
    $$
and let $\tau'\in\Omega_{\partial W}$ be the restriction of $\tau$ to $\partial W$.  Then 
\begin{equation}\label{eq:Gibbs}
\mu_W^{(V\setminus W,\tau)}=\mu_W^{(\partial W,\tau')};
\end{equation}
in other words, conditioning on the complement of $W$ is equivalent to conditioning on its boundary.

As we are claiming that our algorithm samples from the Gibbs measure on an infinite graph, namely the 2-D grid, we need to say a little about what that means.  A full treatment of `infinite volume' measures is given by Friedli and Velenik~\cite{friedli_velenik_2017}: see Section 3 and in particularly Subsection 3.3.  The bare bones are as follows.

If we drop the assumption that $V$ is finite, identity~(\ref{eq:Gibbs}) still enables us to calculate the marginal distribution $\mu_W^{(V\setminus W,\tau)}$ provided $W$ is finite.  The quantities 
$$\big\{\mu_W^{(V\setminus W,\tau)}:\text{$W\subset V$, $W$ finite, and $\tau\in\Omega_{V\setminus W}$}\big\}$$
provide a \emph{specification} of the infinite volume measure~\cite[Def.~6.11]{friedli_velenik_2017}.  (The various marginal distributions have to satisfy a consistency condition, but that is automatic here.)  The infinite volume measures we are interested in are ones that are consistent with the specification~\cite[Defn 6.12]{friedli_velenik_2017}.  There may be no, one or many such measures.  In our algorithm we only make use of marginal probabilities on finite subgraphs, which are known by~(\ref{eq:Gibbs}).  Therefore, a proof that our algorithm terminates with probability~1 can be interpreted as a proof that the infinite volume measure is unique, and the output of the algorithm can be interpreted as a construction of (the marginal distribution on a finite subset $U\subset V$ of) the infinite volume measure.

Finally, we say $A$ is a \emph{soft} interaction matrix if there exists $\alpha > 0$ and $i \in [q]$ such that for all $j \in [q], A_{ij} > \alpha$. In this paper we restrict ourselves to spin systems with soft interaction matrices. It is unnecessarily restrictive, but straightforward to handle and enough for our applications.

\subsection{Weak Spatial Mixing}
Note that in the following $d_G$ denotes graph distance. 

For efficiency of our algorithm we require weak spatial mixing, which intuitively tells us that, given sets $U \subset \bZ^2$ and $W \subset \bZ^2$, conditioning on $\sigma_W$ has little effect on the distribution of $\sigma_U$ so long as $W$ is far from $U$.  In order to connect smoothly with previous work, we use essentially the same definition as Martinelli, Olivieri and Schonmann~\cite[(1.11)]{Martinelli1994weakstrong}.  

\begin{definition}\label{def:wsm}
Let $\mathcal S$ be a spin system and $f : \mathbb N \to \mathbb R^+$ a function.  Given $W\subset V, U \subset V \setminus W,$ and $\tau^1,\tau^2\in\Omega_W$, if for all $W$ with $V\setminus W$ finite, all $U\subset V\setminus W$, and all feasible $\tau^1, \tau^2 \in \Omega_W$ with $\tau^1\not=\tau^2$,
\begin{align*}
  d_\mathrm{TV}\left(\mu_U^{\tau^1},\mu_U^{\tau^2} \right) \leq \sum_{u \in U, w \in W} f(d_G(u,w)),
\end{align*}
then we say that $\mathcal S$ exhibits \emph{set weak spatial mixing} with rate~$f$.
\end{definition}
From now on, we will drop the subscript $G$ in the graph distance $d_G$ as we are always working on $\mathbb Z^2$.

There are many variants of weak spatial mixing and it is critical we use a strong enough definition. To obtain mixing $O(n \log n)$ mixing time results of the random cluster model on $\bZ^d$, Gheissari and Sinclair require that $U, W$ be boxes \cite{gheissarisinclair2022}. In contrast, we need our sets to take any form. In Spinka's perfect sampling paper \cite{spinka2020finitary}, rather than our sum, the bound is 
\begin{align*}
  d_\mathrm{TV}\left(\mu_U^{\tau^1},\mu_U^{\tau^2} \right) \leq |U| f(d(U,W)),
\end{align*}
however in our case, this condition is not strong enough to ensure efficiency.

From this point on, weak spatial mixing will be with exponential rate, i.e., $f(x) = Ce^{-\gamma x}$.

\section{Martinelli, Olivieri, and Schonmann}
In their 1994 paper \cite{Martinelli1994weakstrong}, obtained a particular result on $\mathbb Z^2$ which showed that, for 2-d spin systems, weak spatial mixing implies strong spatial mixing for squares. Specifically, the result is as follows. Let $S \subset \bZ^2$ be a square of side length $L$, that is, for some vertex~$v_c$,
\begin{align*}
    S = \set{v \in \bZ^2 : v = v_c + (a,b), 0 \leq a,b \leq L}.
\end{align*}

\begin{theorem}[{\cite[Theorem 1.1]{Martinelli1994weakstrong}}]\label{MOStheorem}
    Let $\mathcal S$ be a spin system on $\bZ^2$ which exhibits weak spatial mixing. Then there exist constants $C, \gamma > 0$ such that for all $\tau_1, \tau_2 \in \Omega_{\partial S}$ and all $\Delta \subset S$
    \begin{align*}
        d_\mathrm{TV}\bra{\mu_\Delta^{\tau_1}, \mu_\Delta^{\tau_2}} \leq C e^{- \gamma \ell},
    \end{align*}
    where $\ell = \min \set{d(v,\Delta) : v \in \partial S, \tau_1(v) \not= \tau_2(v)}.$
\end{theorem}
Here, $d(v,\Delta)=\min\{d(v,u):u\in\Delta\}$.  Martinelli et al.\  also note that this result is not specific to $\bZ^2$ --- it can be extended to other regular 2-d graphs --- however there are counterexamples to show that the restriction to 2 dimensions is essential~\cite{MR1269387}. 

We will not need Theorem \ref{MOStheorem} for the purposes of our algorithm; we only need to adapt the intermediate Proposition~2.1 from the same paper \cite{Martinelli1994weakstrong}. We cannot directly apply the proposition, because our algorithm will construct partial configurations which do not restrict themselves to squares around vertices. 

Consider a square $S \subset \mathbb Z^2$ of side length $2L$ (thus made of 4 squares of side length $L$). To be precise, for some $v_c \in \bZ^2$,
\begin{align*}
    S &= \set{v \in \bZ^2 : v = v_c + (a,b), 0 < a,b < 2L}, \\
    S_0 &= \set{v \in \bZ^2 : v = v_c + (a,b), 0 < a,b < L}, \\
    S_1 &= \set{v \in \bZ^2 : v = v_c + (L + a,b), 0 < a,b < L}, \\
    S_2 &= \set{v \in \bZ^2 : v = v_c + (a,L + b), 0 < a,b < L}, \\
    S_3 &= \set{v \in \bZ^2 : v = v_c + (L + a,L + b), 0 < a,b < L}.
\end{align*}
The notation $\partial S$ will denote the usual boundary of~$S$, and we will also consider the `bisected' boundary of~$S$,
\begin{align*}
    \boxplus S = \partial S_0 \cup \partial S_1 \cup \partial S_2 \cup \partial S_3.
\end{align*}

In analysing our algorithm, we will encounter boundary conditions that fix spins on the whole of $\partial S$ and on a subset of $\boxplus S\setminus \partial S$. Such a boundary condition is specified by a partial configuration $(T,\tau)$, where $\partial S\subseteq T\subseteq \boxplus S$ and $\tau\in\Omega_T$.

\begin{prop} \label{weakprop}
Let $\mathcal S$ be a $q$-spin system on $\mathbb Z^2$ which exhibits weak spatial mixing with constants $C, \gamma > 0$. Then there exist $L_0 \in \mathbb N$ and $\gamma_0 > 0$ such that the following holds:  for all squares $S \subset \mathbb Z^2$ with side length $2L > 2L_0$, all $y \in \partial S$, all $T$ with $\partial S\subseteq T\subseteq \boxplus S$, and all feasible partial configurations $\tau, \tau'\in\Omega_T$ that agree except at vertex~$y$,
\begin{align*}
    d_\mathrm{TV} \left(\mu_{Q}^\tau, \mu_{Q}^{\tau'}\right) \leq e^{- \gamma_0 L^{1/4}},
\end{align*}
where
\begin{align*}
    Q := \set{x \in S : \|x - y\| \geq L^\frac12}.
\end{align*}
\end{prop}

Before proving the result, we run through a simple 1-d analogue of the proof.  Recall that for two measures $\lambda$ and $\mu$ on the same state space and random variables $X\sim\lambda$ and $Y\sim\mu$, if we can find a coupling $(X,Y)$ of the random variables such that $\bP[X \not= Y] \leq \epsilon$, then
\begin{align*}
    d_\mathrm{TV} (\mu, \lambda) \leq \epsilon.
\end{align*}
It follows that to bound the total variation distance between two measures, we need only find a coupling which agrees to an appropriate extent.  For the convenience of readers who are familiar with~\cite{Martinelli1994weakstrong}, we adopt some of the notation used there.  Note that there is no particular significance in the quantities $L$, $L^{\frac12}$ and $L^{\frac14}$ beyond the fact that the first grows more quickly than the second, and the second more quickly than the third.


Consider a 2-spin system $\mathcal S$ on $S=\{-L,...,L\} \subset \bZ$, with $L^{\frac14}$ integral, which exhibits weak spatial mixing with constants $C, \gamma$. Suppose we want to determine the influence of a fixed spin at 0 on the set 
\begin{align*}
    Q := \set{x \in S : \abs{x} \geq L^\frac12}.
\end{align*}
We can do so by building two configurations on $S$ which we attempt to couple at each step; that is, we try to ensure that the two configurations are equal sufficiently far from~0.  We denote the two configurations by $\sigma^+$ and $\sigma^-$, with boundary conditions  $\sigma^+_{-L}=\sigma^-_{-L}={+}$, $\sigma^+_L=\sigma^-_L={+}$, and $\sigma^+_0={+}$ and $\sigma^-_0={-}$. We sample $\sigma^+$ and $\sigma^-$ spin-by-spin, starting at the origin, trying to make them agree, while at the same time ensuring that both $\sigma^+$ and $\sigma^-$ have the correct Gibbs distribution.  If we have ever succeeded in coupling the two configurations then the Markov property ensures we can always maintain our coupling. Thus we only need to analyse the case where we have not already coupled.


We will attempt to couple every $L^{\frac14}$ spins. Let $D_i=\{-iL^{\frac14},iL^{\frac14}\}$.  Suppose that $\sigma^+_{D_{i-1}}\not=\sigma^-_{D_{i-1}}$, i.e., we have not coupled at the previous step.  Since $\cS$ exhibits weak spatial mixing, we can couple $\sigma^+_{D_i}$ with $\sigma^-_{D_i}$ so that
\begin{align*}
    \bP \sbra{\sigma^+_{D_i} \not= \sigma^-_{D_i}} < C e^{- \gamma L^{1/4}}, 
\end{align*}
By iterating this procedure, it follows that at step $n=L^{\frac14}$, the probability we have not coupled is bounded:
\begin{align*}
    \bP \sbra{\sigma^+_{D_n} \not= \sigma^-_{D_n}} < C^ne^{- \gamma nL^{1/4}}=C^{L^{1/4}}e^{- \gamma L^{1/2}}.
\end{align*}
By choosing $L$ sufficiently large, we can ensure that coupling has taken place with probability at least $\frac12$, in which case we can make $\sigma^+$ and $\sigma^-$ agree on~$Q$.  The following proof for the 2-d situation is more involved, but is similar.

\begin{proof}[Proof of Proposition~\ref{weakprop}]
Assume again that $L^\frac14$ is integral. Taking our cue from the warm-up 1-d argument,  we will build a sequence of frames $D_i \subset S$ around $y$, each $L^\frac14$ further from $y$ than the previous. We will then construct configurations $\sigma, \sigma'$, starting with $\tau$ and $\tau'$ respectively, and with each of these frames we will attempt to couple the two configurations. The probability of us coupling will give the bound on $d_\mathrm{TV} \bra{\mu_{Q}^\tau, \mu_{Q}^{\tau'}}$.

Let $h$ be a sufficiently large number, to be chosen later.  Define, for $i = 0,...,L^\frac14$:
\begin{align*}
    Q_i &= \set{x \in  S : \|x - y\|_\infty > iL^\frac14}, \\
    D_i &= \set{x \in  S : \|x - y\|_\infty = iL^\frac14}, \\
    A_i &= \set{x \in D_i : \min_{z \in \boxplus S} \|x - z\|_\infty \leq h }, \\
    B_i &= D_i \setminus A_i.
\end{align*}

\begin{figure}[h]
\centering
\includegraphics[width=0.75\textwidth]{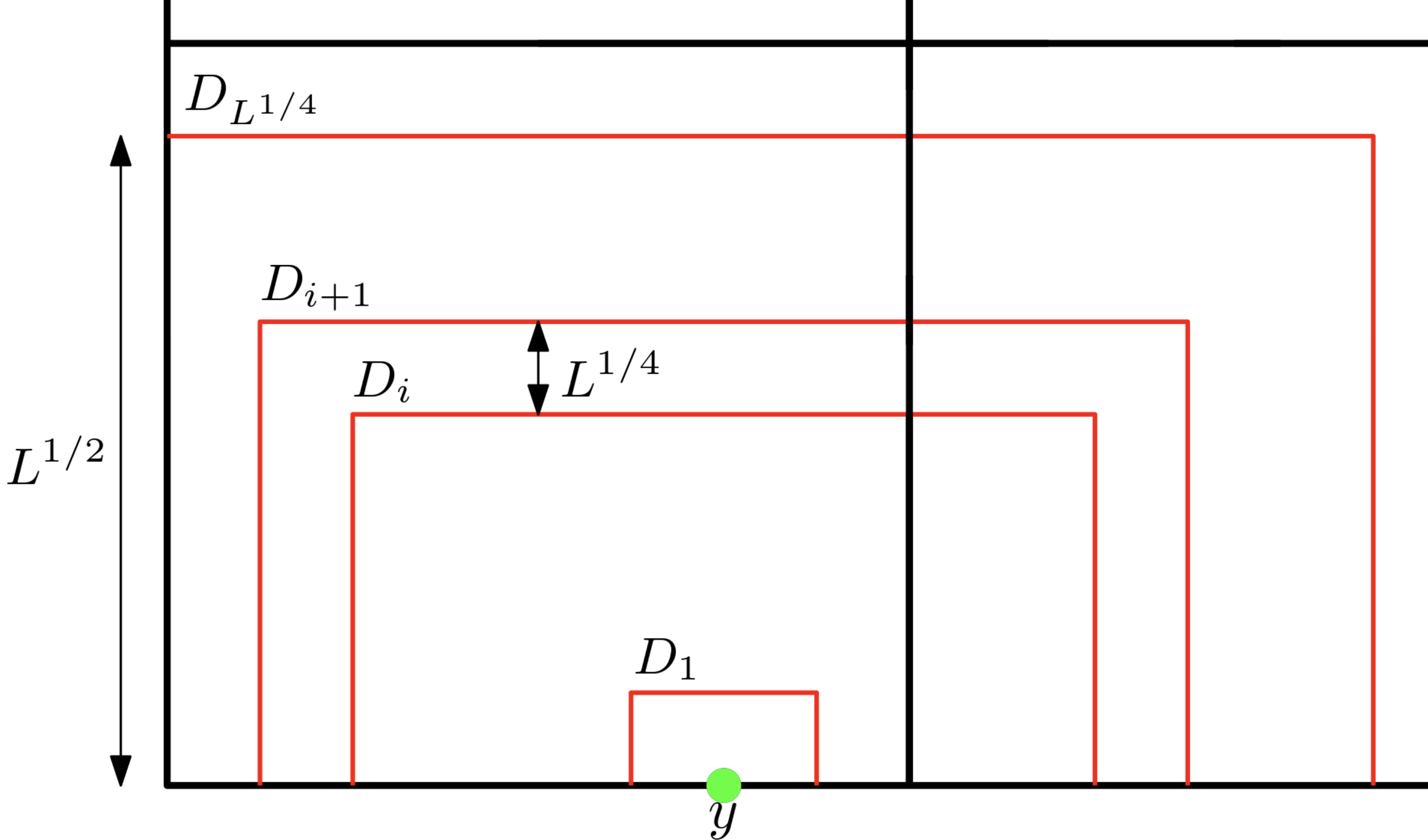}
\caption{In black $\boxplus S$, the domain of $\tau$ and $\tau'$. In green the vertex $y$ on which $\tau$ and $\tau'$ differ. In red the $D_i$'s, i.e. the $L^\frac14$ frames on which we will iteratively attempt to couple $\sigma$ and $\sigma'$. Disclaimer: this image is indicative, not to scale.}
\end{figure}

The $D_i$'s will be our frames. They are split into the $A_i$'s, the pieces close to the boundary, and the $B_i$'s, the pieces far from the boundary. The $B_i$'s are the nice vertices which are far enough from $\boxplus S$ that weak spatial mixing applies. (The parameter $h$ will be chosen so that is the case.)  The $A_i$'s are the unpleasant vertices, for which we have no nice trick, but which are few enough in number that we can deal with them quite crudely.  We will attempt to couple these two sets separately. Note that, again, once we have matched the configurations on the $i$th frame, we can always extend the coupling to the whole of $Q_i$ due to the Markov property and the identical boundaries of $Q_i$ on both configurations. Thus, we will only discuss the case where we have failed at every previous step.

Consider the sets $A_i$ and $B_i$. For a particular $h$, we may choose $L_0$ large enough so that $L_0^\frac14 > 2h$. Then for all but at most one $i_0$, the $A_i$'s will be made of (at most) four `feet' of size $h$ which are the only pieces of the $D_i$'s  close to $\boxplus S$. We will skip the attempt at coalescence for $i_0$ and we will still be able to succeed with probability arbitrarily close to one.

Now, we attempt to couple by first dealing with the $B_i$'s and then the $A_i$'s. For the $B_i$'s, we will need the following technical lemma, which tells us that the influence on $B_i$ of the fixed partial configuration is bounded, independently of $h$ and~$L$, so long as $h$ and $L$ are chosen large enough.

\begin{figure}[h]
\centering
\includegraphics[width=0.9\textwidth]{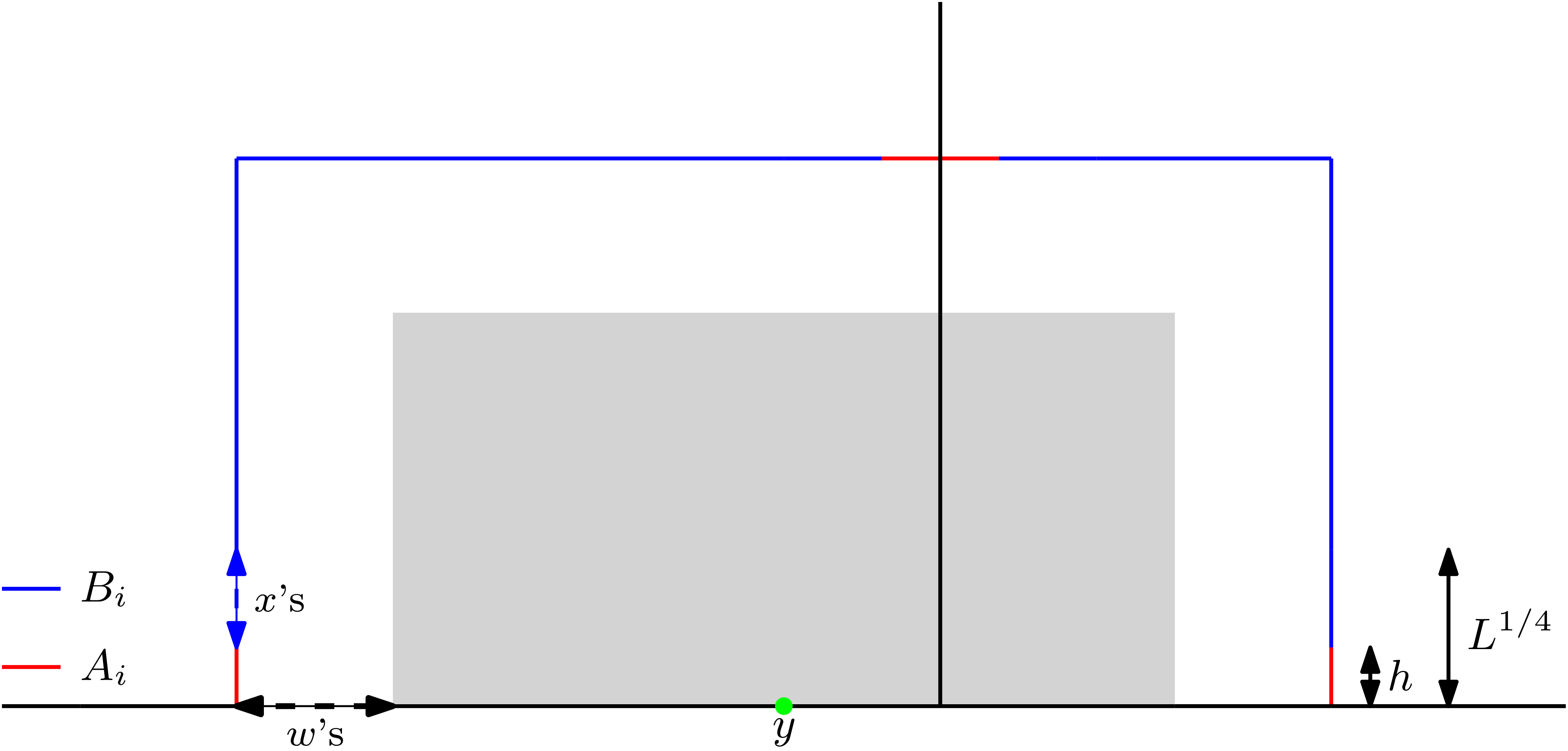}
\caption{In black, the vertices in $\boxplus S$. In grey, $Q_{i-1}^c$, i.e. the vertices which have been fixed once we reach the $i$th frame. In red, $A_i$, the vertices of $D_i$ which are too close to the boundary to leverage weak spatial mixing. In blue, $B_i$, the vertices in $D_i$ which are far from $\boxplus S$. Marked by arrows, the $x's$ and $w$'s which form $E$. The $x$'s are the vertices in $B_i$ that are within $L^\frac14$ of $\boxplus S$ and the $w$'s are the vertices in $\boxplus S$ which are within $L^\frac14$ of $D_i$.}
\label{fig:oneiter}
\end{figure}

\begin{lemma} \label{simplefact}
Under the conditions of Proposition~\ref{weakprop}, there are $h, L_0$ such that for all $L > L_0$, and for all $i = 1,...,L^\frac14$ with $ i \not= i_0$:
\begin{align*}
    C\!\!\! \sum_{\substack{x \in B_i, \\ w \in (\boxplus S \cup D_{i-1})}} e^{-\gamma d(x,w)} \leq \frac12.
\end{align*}
This will give us a coupling of $\mu_{B_i}^{\tau \oplus \sigma_{D_{i-1}}}$ and $\mu_{B_i}^{\tau' \oplus \sigma'_{D_{i-1}}}$ which agrees with probability greater than $\frac12$.
\end{lemma}

\begin{proof}
We split the sum into two pieces.  Let
\begin{align*}
    P &= \set{(x,w) \in B_i \times (\boxplus S \cup D_{i-1}) : d(x,w) \leq L^\frac14}\\
    \noalign{\noindent and}
    P^c &= B_i \times (\boxplus S \cup D_{i-1})\setminus P.
\end{align*}
Then
\begin{align*}
    \sum_{\substack{x \in B_i, \\ w \in (\boxplus S \cup D_{i-1})}} e^{-\gamma d(x,w)} =  \sum_{(x,w) \in P} e^{-\gamma d(x,w)} + \sum_{(x,w) \in P^c} e^{-\gamma d(x,w)}.
\end{align*}

Note that $D_i$ intersects $\partial S$ twice, and $\boxplus S\setminus\partial S$ at most once.  So $P$ can be partitioned into at most three subsets, each local to some intersection.  Let $z$ be one of the intersections between $D_i$ and $\boxplus S$.  We can further partition the pairs $(x,w)$ in $P$ according to whether $x$ and $w$ lie left/right or above/below~$z$.  In this way, we partition $P$ into at most eight (broken) L-shaped pieces $\{P_1,\ldots,P_8\}$, each local to some intersection.  Let $P_m$ be such an L-shaped region local to $z$.  (Refer to Figure~\ref{fig:oneiter}.)  Then we have

$$
\sum_{(x,w) \in P_m} e^{-\gamma d(x,w)}
=\sum_{(x,w) \in P_m} e^{-\gamma(d(x,z)+d(w,z))}
\leq \sum_{j=0}^\infty \sum_{k=0}^\infty e^{-\gamma (j+k + h)} 
= \frac{e^{-\gamma h}}{(1 - e^{-\gamma})^2}.
$$
Here, the summations run over $j=d(x,z)-h$ and $k=d(w,z)$.  Thus 
$$
    C \sum_{(x,w) \in P} e^{-\gamma d(x,w)} = C \sum_{m=1}^8\sum_{(x,w) \in P_m} e^{-\gamma d(x,w)}
    \leq \frac{8C e^{-\gamma h}}{(1 - e^{-\gamma})^2}.
$$
We can take $h$ large enough to bound the above by $\frac1{4}.$

For $P^c$ we note that every term in the sum has $d(x,w) \geq L^\frac14$ and the number of terms is bounded by $13L^2$, so
\begin{align*}
    C \sum_{(x,w) \in P^c} e^{-\gamma d(x,w)} < 13 C L^2 e^{-\gamma L^{1/4}}.
\end{align*}
Then taking $L_0$ large enough, we can bound this term by $\frac14$.

Now with our choice of $h$ and any $L > L_0$ we have that
\begin{align*}
    C\!\! \sum_{\substack{x \in B_i, \\ w \in (\boxplus S \cup D_{i-1})}} e^{-\gamma d(x,w)} < \frac1{4} + \frac14 = \frac12.
\end{align*}
This completes the proof of Lemma~\ref{simplefact}.
\end{proof}

Now we describe the iterated coupling attempts. Suppose at step $i$ we have thus far failed to couple $\sigma$ and $\sigma'$ outside $D_{i-1}$.

For the $B_i$'s, it is critical that we have our definition of weak spatial mixing, since it gives us that for every $i = 1,...,L^\frac14$,
\begin{align*}
    d_\mathrm{TV}\left(\mu_{B_i}^{\tau \oplus \sigma_{D_{i-1}}}, \mu_{B_i}^{\tau' \oplus \sigma'_{D_{i-1}}} \right) < C\!\! \sum_{\substack{x \in B_i, \\ w \in (\partial S \cup B_{i-1})}} e^{-\gamma d(x,y)}.
\end{align*}
Then taking $h$ and $L_0$ from \ref{simplefact}, it follows immediately that
\begin{align*}
    d_\mathrm{TV}\left(\mu_{B_i}^\tau, \mu_{B_i}^{\tau'} \right) < \frac12.
\end{align*}
It follows that 
we can choose a coupling on $B_i$ so that
\begin{align*}
    \mathbb P \left[\sigma_{B_i} = \sigma'_{B_i} \right] \geq \frac12.
\end{align*}

For the $A_i$, we are too close to use weak spatial mixing, so we will use a cruder method, sampling $\sigma_{A_i}$ and $\sigma'_{A_i}$ independently. Due to our restriction to soft interaction matrices, each time we attempt to couple, we have for some constant $\alpha>0$,
\begin{align*}
    \mathbb P \left[\sigma_{A_i} = \sigma'_{A_i} \right] \geq \alpha^{|A_i|} = \alpha^{4h}.
\end{align*}

Now, combining the two pieces, we find that on every iteration where we have not succeeded, we couple with probability
\begin{align*}
    \mathbb P \left[\sigma_{D_i} = \sigma'_{D_i} \right] \geq \frac12 \alpha^{4h}.
\end{align*}
It follows that
\begin{align*}
    d_\mathrm{TV} \left(\mu_{Q}^\tau, \mu_{Q}^{\tau'}\right) \leq \mathbb P \left[\sigma_{D_{L^{1/4}}} = \sigma'_{D_{L^{1/4}}} \right] < \left(1 -  \frac12 \alpha^{4h} \right)^{L^{1/4}} < e^{- \frac12 \alpha^{4h} L^{1/4}}.
\end{align*}
Which completes the proof of Lemma~\ref{simplefact}.
\end{proof}

We stress again that requiring soft interaction matrices is unnecessarily restrictive. It is simple and satisfactory for the chosen applications, but so long as the spin system has some constant probability of coupling $\sigma_{A_i}$ and $\sigma'_{A_i}$, the argument still holds.

For the purposes of the algorithm we will not have two configurations that differ on only one vertex, but luckily we can fix this easily. First define
\begin{align*}
    \Delta := \set{x \in S : \forall y \in \partial S \textrm{ such that } \tau_y \not= \tau_y', \|x - y\| \geq L^\frac12}.
\end{align*}

\begin{corollary} \label{weakcor}
Let $\mathcal S$ be a $q$-spin system on the integer lattice $\mathbb Z^2$ which exhibits weak spatial mixing. Then there exists $L_0 \in \mathbb N$ and a positive constant $\gamma_0$ such that for all squares $S \subset \mathbb Z^2$ with side length $2L > 2L_0$, all sets $T$ with $\partial S\subseteq T\subseteq\boxplus S$, and all  configurations $\tau, \tau' \in \Omega_T$ which agree on $\boxplus S\setminus\partial S$,
\begin{align*}
    d_\mathrm{TV} \left(\mu_{\Delta}^\tau, \mu_{\Delta}^{\tau'}\right) \leq 8L e^{- \gamma_0 L^{1/4}}.
\end{align*}
\end{corollary}

\begin{proof}
    We simply change $\tau$ to $\tau'$ by flipping the spin at one boundary vertex at a time. There are $8L$ boundary vertices, so applying the triangle inequality and Proposition \ref{weakprop} gives the result.
\end{proof}

\section{The algorithm}

This version of lazy depth-first sampling is very similar to the the one we present in \cite{AnandJerrum2021}, but here we take more care in our choices of vertices to sample.

Before we begin the algorithm, we define a mesh $\Gamma \subset \mathbb Z^2$, where each mesh square in $\Gamma$ has side-length $L$. To be precise, $\Gamma = (V_\Gamma, E_\Gamma)$ is defined with
\begin{align*}
    V_\Gamma &= \set{v = (i,j) \in \bZ^2 : i = 0 \pmod L} \cup \set{v = (i,j) \in \bZ^2 : j = 0 \pmod L}, \\
    E_\Gamma &= \set{(u,v): u,v\in V_\Gamma, d(u,v) = 1}.
\end{align*}
To ensure that the bound on total variation distance from Proposition~\ref{weakprop} applies, \textsc{lazy} will only sample vertices in~$\Gamma$. Also, for each $v = (i,j) \in V_\Gamma^2$, define $\Gamma_v \subset \Gamma$ to be the square of side length $2L$ which contains $v$ and for which $d\bra{\Gamma_v, v} \geq L/2.$ Precisely, let $v_c = (k,\ell) \in \Gamma$ be the vertex such that $k,\ell = 0 \pmod L$ and
\begin{align*}
    k + \frac L2 \leq i < k + \frac{3L}2, \\
    \ell + \frac L2 \leq j < \ell + \frac{3L}2.
\end{align*}
and define
\begin{align*}
    \Gamma_v = \partial \set{u \in \bZ^2 : u = v_c + (a,b), 0 < a,b < 2L}.
\end{align*}
When we recurse in the algorithm called on vertex $v$, we will always recurse on vertices in~$\Gamma_v$.


The algorithm takes a $q$-spin system on $\mathbb Z^2$ with a partial configuration $\sigma$ on some subset of the vertices, $\Lambda \subset \Gamma$. 
Informally, the partial configuration $(\Lambda,\sigma)$ represents the spins that have already been determined.  The algorithm is then called on some unassigned vertex~$v \in \Gamma$.  We compute the following probabilities:
\begin{align*}  \label{ps}
  p_v^i &:= \min_{\tau \in \Omega_{\Gamma_v \setminus \Lambda}} \mu_v^{(\Lambda,\sigma)\oplus(\Gamma_v\setminus \Lambda,\tau)}(i) \quad \forall i \in \{1,...,q\}, \\
  p_v^0 &:= 1 - \sum_{i \in [q]} p_v^i.
\end{align*}
With these values, we partition the unit interval into $q+1$ subintervals of lengths $\bra{p_v^i:0\leq i\leq q}$, and draw a sample $y \sim U[0,1]$. If $y$ is in a subinterval corresponding to $p_v^i$, with $i\not=0$, then we set the spin to $i$---if not, we recursively call the algorithm on the vertices in $\Gamma_v\setminus \Lambda$ in turn, each time conditioning on the assignments to the previous vertices in this set. In the case that we need recursion, we use the subroutine \textsc{bd-calc} to compute the spin at $v$. (See Algorithm~\ref{alg:LDFS}.)

Note that this differs from the vanilla \textsc{lazy} in \cite{AnandJerrum2021} (there referred to as \textsc{ssms}) by our choice of recursion: here, to take advantage of the total variation bound we must restrict our samples to a grid $\Gamma$, while on a spin system which exhibits strong spatial mixing we can simply look at the ball of radius~$L$ around~$v$ instead.

Note also that a spin system is determined by a graph, a number of spins, a vector of weights, and an interaction matrix. Our algorithm is specifically for the case of $q$ spins on the graph $\mathbb Z^2$, so we include these in the statement of the algorithm for emphasis.

\begin{algorithm}
\caption{$\textsc{lazy}(\mathcal S = (\mathbb Z^2, q, b, A), (\Lambda, \sigma), v, \Gamma)$}\label{alg:LDFS}
\begin{algorithmic}[1]
\STATE{\textbf{Input:}~The algorithm takes a spin system $\mathcal S$, a grid $\Gamma \subset \mathbb Z^2$, a set of known vertices $\Lambda \subset \Gamma$ with a configuration $\sigma \in \Omega_\Lambda$, and a vertex to sample $v \in \Gamma\setminus\Lambda$.}
\STATE{\textbf{Output:}~The algorithm returns the partial configuration passed in with a spin at $v$ as well: $(\Lambda, \sigma) \oplus (v,i)$ for some $i \in \{1,...,q\}$.}
\STATE{Define $\Gamma_v$ as above.}

\FOR{$i \in \{1,...,q\}$}
  \STATE{$p_v^i \leftarrow \min_{\tau \in \Omega_{\Gamma_v\setminus\Lambda}} \mu_v^{\sigma \oplus \tau}(i)$}
  \STATE{$I_i \leftarrow \left[\sum_{j = 1}^{i-1} p_v^j, \sum_{j = 1}^i p_v^j\right)$}
\ENDFOR

\STATE{ $p_v^0 \leftarrow 1 - \sum_{i \in [q]} p_v^i$}
\STATE{$I_0 \leftarrow \left[1-p_v^0,1\right]$}
\STATE{Sample $y \sim U[0,1]$, a realisation of a uniform $[0,1]$ random variable.}
\FOR{$i = \{1,...,q\}$}
  \IF{$y \in I_i$}
    \RETURN $((\Lambda, \sigma) \oplus (v,i))$
  \ENDIF
\ENDFOR
\RETURN \textsc{bd-calc} $\bra{\mathcal S, (\Lambda, \sigma), v, (p_v^1,\ldots,p_v^q),y}$

\end{algorithmic}
\end{algorithm}

\begin{algorithm}
\caption{\textsc{bd-calc} $(\mathcal S, (\Lambda, \sigma), v, (p_v^1,\ldots,p_v^q),y)$}\label{alg:BD-calc}
\begin{algorithmic}[1]
\STATE{Give $\Gamma_v\setminus\Lambda$ an ordering $\Gamma_v\setminus\Lambda = \{w_1,...,w_m\}$.}
\STATE{$(\Lambda',\sigma') \leftarrow (\Lambda,\sigma)$}

\FOR{$j \in [1,...,m]$}
  \STATE{$(\Lambda',\sigma') \leftarrow $\textsc{lazy}$(\mathcal S, (\Lambda',\sigma'), w_j, \Gamma)$}
\ENDFOR

\FOR{$i \in \{1,...,q\}$}
  \STATE{$\rho_v^i \leftarrow \mu_v^{\sigma'}(i) - p_v^i$}
  \STATE{$I_i \leftarrow \left[\sum_{j = 1}^{q} p_v^j + \sum_{k = 1}^{i-1}\rho_v^k, \sum_{j = 1}^{q} p_v^j + \sum_{k = 1}^{i}\rho_v^k \right)$}
  \IF{$y \in I_i$}
    \RETURN $((\Lambda, \sigma) \oplus (v,i))$
  \ENDIF
\ENDFOR
\end{algorithmic}
\end{algorithm}

It is worth noting that this algorithm also provides a way to sample vertices outside the mesh $\Gamma$. For $v = (i,j) \in \bZ^2 \setminus \Gamma$, let $v_c = (k,\ell) \in \Gamma$ be the vertex such that $k,\ell =0\pmod L$ and
\begin{align*}
    k < i < k + L, \\
    \ell < j < \ell + L,
\end{align*}
and define
\begin{align*}
    B &= \partial \set{u \in \bZ^2 : u = v_c + (a,b), 0 < a,b < L}.
\end{align*}
The algorithm allows us to sample the spin of every vertex in~$B$, at which point we can sample the spin at~$v$ by using any algorithm to sample spin systems with boundary.

The verification of correctness of \textsc{Lazy} is similar to that presented in \cite{AnandJerrum2021} as is  the analysis of its run-time. The run-time analysis differs more, so we treat it first.

\section{Run-time}

Our run-time analysis is in terms of the number of calls of \textsc{lazy}. For properly chosen $L$, and hence $\Gamma$, the amount of time each call takes for vertices on $\Gamma$ will be $O(1)$.

Our viewpoint is that the algorithm behaves similarly to a branching process where each recursive call of the algorithm is a child of the vertex which initiated the call. This gives us an upper bound to the run-time of \textsc{lazy} on $\mathbb Z^2$ in a broad class of spin systems.

Now note that, although the interior~$S$ of the square~$\Gamma_v$ of $\Gamma$ may not be empty as the algorithm progresses, only the vertical and horizontal bisectors of $\Gamma_v$ may be occupied.  In other words, only the spins of vertices in 
$\Gamma$ have been fixed.  From our adaptation of Martinelli, Olivieri, and Schonmann's result (Corollary~\ref{weakcor}), we know that 
the distribution of the spin at $v$ is only weakly influenced by the boundary condition on $\Gamma_v$

\begin{theorem} \label{runtime}
	Let $\mathcal S = \bra{\mathbb Z^2,q,b,A}$ be a spin system with a soft interaction matrix $A$, and which exhibits weak spatial mixing. Take $\gamma_0$ and $L_0$ from Corollary \ref{weakcor}, and let $L > L_0$ satisfy
    \begin{align*}
        64L^2 e^{-\gamma_0 L^{1/4}} < 1.
    \end{align*}
    Let $\Gamma$ be an $L$-mesh. Then for all $v \in \Gamma$, $\Lambda \subset \Gamma$, and $\sigma \in \Omega_\Lambda$ such that
    \begin{align*}
      \mu_\Lambda(\sigma) > 0,
    \end{align*}
    the expected number of times \textsc{lazy}$(\mathcal S, (\Lambda,\sigma), v, \Gamma)$ calls \textsc{lazy} (in total) is $O(1)$.
\end{theorem}

Note that we can always choose $L$ large enough to satisfy $L > L_0$ and $64L^2 e^{-\gamma_0 L^{1/4}} < 1$.

\begin{proof}
	We construct a random tree $T_A$ out of our recursive calls. Our root will be our initial call. Recursively, if at some node $w$ of the tree we fall in the zone of indecision, $I_0$, we add a child to the node for every recursive call that is initiated by \textsc{bd-split}. If we do not fall in the zone of indecision, then $w$ has no children.
	
	Let $T_B$ be the branching process with offspring distribution $\xi$ where
	\begin{align*}
      \mathbb P[\xi = 0] = 1 - 8L e^{- \gamma_0 L^{1/4}}, \quad \mathbb P[\xi = 8L] = 8L e^{- \gamma_0 L^{1/4}}.
    \end{align*}
    Then we can couple $T_B$ with $T_A$ so that $T_B$ stochastically dominates $T_A$. Indeed, Corollary \ref{weakcor} tells us that at every node $w \in T_A$, the probability of falling in the zone of indecision is bounded by $8L e^{- \gamma_0 L^{1/4}}$. Further, when $w$ does have children, the number of children is bounded above by $|\Gamma_v| = 8L$.
    
    Now we can leverage the theory of branching processes to finish off our bound. We know that
    \begin{align*}
      \mathbb E [\xi] \leq 64L^2 e^{-\gamma_0 L^{1/4}} < 1,
    \end{align*}
    so by the fundamental theorem of branching processes, it follows that
    \begin{align*}
    	\mathbb E[|T_A|] \leq \mathbb E[|T_B|] \leq \frac1{1 - 64L^2 e^{-\gamma_0 L^{1/4}}} = O(1). 
    \end{align*}
\end{proof}


The qualitative nature of the main result in this section is summarised in the following corollary.

\begin{corollary} \label{general}
    Let $\mathcal S = (\mathbb Z^2,q,b,A)$ be a spin system which has a soft interaction matrix~$A$, and which exhibits weak spatial mixing. Let $\Gamma \subset \mathbb Z^2$ be an $L$-grid with $L$ chosen large enough to satisfy Theorem~\ref{runtime} and consider any choices of $v \in V$, $\Lambda \subset \Gamma$, and $\sigma \in \Omega_\Lambda$. Then the running time $R$ of \textsc{lazy}$(\mathcal S, (\Lambda,\sigma), v, \Gamma)$ satisfies
    $\mathbb E[R] = O(1)$.
\end{corollary}

\section{Correctness}

The proof of correctness is nearly identical to that in \cite{AnandJerrum2021}, so we include only an intuitive explanation.

\begin{theorem} \label{correctness}
	Let $\mathcal S = (\bZ^2,q,b,A)$ be a spin system which ahs a soft interaction matrix~$A$ and which exhibits weak spatial mixing, and let $\Gamma$ be the chosen mesh. Suppose for all $v \in \Gamma$, all finite $\Lambda \subset \Gamma$, and all $\sigma \in \Omega_\Lambda$ such that
    \begin{align*}
      \mu_\Lambda(\sigma) > 0,
    \end{align*}
    \textsc{lazy}$(\mathcal S, (\Lambda,\sigma), v, \Gamma)$ terminates with probability 1. Then
    \begin{align*}
    	\mathbb P\left[\textsc{lazy}(\mathcal S, (\Lambda,\sigma), v,\Gamma) = i\right] = \mu_v^\sigma(i).
    \end{align*}
\end{theorem}

We can show the algorithm is correct in an inductive manner, just one lacking a traditional base case. The issue here is that the tree of recursive calls has paths of unbounded length.  In the earlier paper~\cite{AnandJerrum2021}, this problem is circumvented by analysing an artificial version of the algorithm in which the recursion depth is restricted to a bound~$h$.  Correctness of the actual algorithm follows by taking a limit as $h\to\infty$, assuming termination with probability~1. The key factor in the algorithm's correctness is that each recursive call samples a spin at a vertex in $\Gamma_v$ from the correct conditional distribution, given the spins already decided.  We thus arrive at the correct distribution of spins on the whole of~$\Gamma_v$, by the chain rule for conditional probabilities. 

Suppose we know that the recursive calls give us the correct probability, that is that
\begin{align*}
    \mathbb P\left[\textsc{lazy}\left(\mathcal S, (\Lambda, \sigma) \oplus \bigoplus_{k = 1}^{j-1}(w_k,\tau_{w_k}),w_j,\Gamma\right) = \tau_{w_j} \right] = \mu_{w_j}^{\left(\Lambda, \sigma) \oplus \bigoplus_{k = 1}^{j-1}(w_k,\tau_{w_k}\right)}(\tau_{w_j}).
\end{align*}
Then given the order $\Gamma_v\setminus\Lambda = \{w_1,...,w_m\}$ that the algorithm processes the vertices in, we have
\begin{align*}
  &\mathbb P\left[\textsc{lazy}(\mathcal S, (\Lambda, \sigma),v,\Gamma) = i \right] \\
  =\ &p_v^i + 
  \sum_{\tau \in \Omega_{\Gamma_v\setminus\Lambda}} \left(\mu_v^{\sigma \oplus \tau}(i) - p_v^i\right) \prod_{j=1}^m \mathbb P\left[\textsc{lazy}\left(\mathcal S, (\Lambda, \sigma) \oplus \bigoplus_{k = 1}^{j-1}(w_k,\tau_{w_k}),w_j,\Gamma\right) = \tau_{w_j} \right] \\
  =\ &p_v^i + \sum_{\tau \in \Omega_{\Gamma_v\setminus\Lambda}} \left(\mu_v^{\sigma \oplus \tau}(i) - p_v^i\right) \prod_{j=1}^m \mu_{w_j}^{\left(\Lambda, \sigma) \oplus \bigoplus_{k = 1}^{j-1}(w_k,\tau_{w_k}\right)}(\tau_{w_j}).
\end{align*}
Now we note that
\begin{align*}
  p_v^i = \sum_{\tau \in \Omega_{\Gamma_v\setminus\Lambda}} p_v^i \cdot \prod_{j=1}^m \mu_{w_j}^{\left(\Lambda, \sigma) \oplus \bigoplus_{k = 1}^{j-1}(w_k,\tau_{w_k}\right)}(\tau_{w_j}),
\end{align*}
from which it follows that
\begin{align*}
  &\mathbb P\left[\textsc{lazy}(\mathcal S, (\Lambda, \sigma),v,\Gamma) = i \right] \\
  =\ &p_v^i + \sum_{\tau \in \Omega_{\Gamma_v\setminus\Lambda}} \left(\mu_v^{\sigma \oplus \tau}(i) - p_v^i\right) \prod_{j=1}^m \mu_{w_j}^{\left(\Lambda, \sigma) \oplus \bigoplus_{k = 1}^{j-1}(w_k,\tau_{w_k}\right)}(\tau_{w_j}) \\
  = &\sum_{\tau \in \Omega_{\Gamma_v\setminus\Lambda}} \mu_v^{\sigma \oplus \tau}(i) \prod_{j=1}^m \mu_{w_j}^{\left(\Lambda, \sigma) \oplus \bigoplus_{k = 1}^{j-1}(w_k,\tau_{w_k}\right)}(\tau_{w_j}) \\
  =\ &\mu_v^\sigma (i).
\end{align*}
Which gives us correctness.

\section{Applications}

\subsection{The ferromagnetic Potts model}
An interesting application of this algorithm is that we can sample the ferromagnetic Potts model on $\mathbb Z^2$ up to its weak spatial mixing threshold. The Potts model can be characterized by the following interactions:
\begin{align*}
  b = 
\begin{pmatrix}
  1 \\
  \vdots \\
  1
\end{pmatrix}, \quad A =
\begin{pmatrix}
  e^{\beta} && 1 && \cdots && 1\\
  1 && e^{\beta} && \cdots && 1\\
  \vdots && \vdots && \ddots && \vdots \\
  1 && 1 && \cdots && e^{\beta}
\end{pmatrix},
\end{align*}
where the number of spins is $q \geq 2$. To be ferromagnetic, $A$ must be positive definite, i.e., $e^{\beta} \geq 1$.

In their seminal paper \cite{MR2948685}, Beffara and Duminil-Copin established for all that the critical inverse temperature of the ferromagnetic Potts model with $q \geq 2$ is
\begin{align*}
    \beta_c = \log \bra{1 + \sqrt q},
\end{align*}
and that correlation decays exponentially for $\beta < \beta_c$.

For our purposes we require not just correlation decay, but weak spatial mixing (of the set form in particular). Fortunately, Alexander showed that for all $\beta < \beta_c$, weak spatial mixing, as we define it here, occurs in the Potts model for $q \geq 2$ if and only if correlation decays \cite[Thm~3.6]{MR1626951}. These two results then ensure that the conditions of Corollary \ref{general} hold and that Algorithm \ref{alg:LDFS} has linear run-time.

It is tempting to approach the low temperature regime via the Edwards-Sokal coupling \cite{edwarssokal} and the dual graph and flows as done by Huijben, Patel, and Regts \cite{vireshguuspotts}. Unfortunately, this correspondence breaks down under the behaviour of our algorithm. To determine spins in the dual, we have to pass through the random cluster model and to do so we must identify connected clusters. If we limit ourselves to boxes of a specific size, what may appear disconnected to us could very well be connected outside our view.

\subsection{The ferromagnetic Ising model with external field}

A second new application is to the ferromagnetic Ising model with constant external field at any temperature. The Ising model with constant external field $h$ on the positive spins can be characterised by

\begin{align*}
  b = 
\begin{pmatrix}
  1 \\
  h
\end{pmatrix}, \quad A =
\begin{pmatrix}
  e^{\beta} && 1 \\
  1 && e^{\beta}
\end{pmatrix}.
\end{align*}

Schonmann and Schlosman established that at all temperatures, the ferromagnetic Ising model with a constant external field exhibits weak spatial mixing \cite[Thm~2]{SchonmannShlosman}. It follows that the conditions of Corollary \ref{general} hold and that Algorithm~\ref{alg:LDFS} has linear run-time on the model.

\printbibliography
\end{document}